%% file: Retrodirective_WET.tex
\documentclass[journal]{IEEEtran}
\usepackage{amsmath}
\usepackage{amssymb}
\usepackage{amsfonts}
\usepackage{graphicx}
\usepackage{epsfig}
\usepackage{subfigure}
\usepackage{psfrag}
\usepackage{color}
\usepackage[noadjust]{cite}

\title{Retrodirective Multi-User Wireless Power Transfer with Massive MIMO
\thanks{The authors are with the Department of Electrical and Computer Engineering, National University of Singapore, Singapore (email: \{s.lee, elezeng, elezhang\}@nus.edu.sg).}
}
\author{Seunghyun Lee,~\IEEEmembership{Student Member,~IEEE}, Yong Zeng,~\IEEEmembership{Member,~IEEE}, and  Rui Zhang,~\IEEEmembership{Fellow,~IEEE} \vspace{-5ex}}

\include{header}

\begin{document}
\maketitle \thispagestyle{empty}

\begin{abstract}
This letter studies a radio-frequency (RF) multi-user wireless power transfer (WPT) system, where an energy transmitter (ET) with a large number of antennas delivers energy wirelessly to multiple distributed  energy receivers (ERs). We investigate a low-complexity WPT scheme based on the \emph{retrodirective} beamforming technique, where all ERs send a common beacon signal simultaneously to the ET in the uplink and the ET simply conjugates and amplifies its received sum-signal and transmits to all ERs in the downlink for WPT. We show that such a low-complexity scheme achieves the massive multiple-input multiple-output (MIMO) energy beamforming gain. However, a ``doubly near-far" issue exists due to the round-trip (uplink beacon and downlink WPT) signal propagation loss where the harvested power of a far ER from the ET can be significantly lower than that of a near ER if the same uplink beacon power is used. To tackle this problem, we propose a distributed uplink beacon power update algorithm, where each ER independently adjusts the beacon power based on its current harvested power in an iterative manner. It is shown that the proposed algorithm converges quickly to a unique fixed-point solution, which helps achieve the desired user fairness with best efforts. 
\end{abstract}


\vspace{-4ex}
\section{Introduction}
Radio-frequency (RF) transmission enabled wireless power transfer (WPT) offers a cost-effective solution for supplying  power perceptually to wireless devices in energy-constrained networks \cite{J_BHZ:2015,J_ZZC:2016}. Compared to alternative WPT technologies such as inductive coupling, RF enabled WPT has several  promising advantages such as wide coverage, low production cost, and smaller transmitter and receiver form factors, etc. In particular, to overcome the significant power loss over distance in RF WPT, the technique of multi-antenna enabled directional transmission or energy beamforming (EB) is essential (see, e.g., \cite{J_ZZC:2016} and the references therein). However, for practical implementation of EB, accurate knowledge of the channel state information (CSI) is required at the energy transmitter (ET). Various CSI acquisition methods designed for WPT have been proposed, including the forward-link training with receiver CSI feedback \cite{J_YHG:2014}, the reverse-link training by exploiting channel reciprocity \cite{J_ZZ:2015,J_ZZ:2015_b}, and the energy-feedback based training \cite{J_XZ:2016,J_LZ:2017,J_CKC:2017}. 

However, the training and feedback overhead of the forward-link training and energy-feedback based training methods can be practically high, especially in multi-user WPT systems with a large number of energy receivers (ERs) and/or massive multiple-input multiple-output (MIMO) WPT systems with a large number of transmit antennas. The reverse-link training approach successfully overcomes this issue as there is no ER feedback required and the training length is independent of the number of transmit antennas, which is an appealing advantage for massive MIMO WPT \cite{J_YHZG:2015,J_KBL:2016}. Nevertheless, the performance of this method critically depends on multi-user orthogonal pilots assignment among the ERs, which causes the so-called ``pilot contamination" problem in massive MIMO WPT systems with many ETs and ERs \cite{J_ZWZ:2016}. Furthermore, multi-user WPT systems in general suffer from the energy near-far problem  \cite{J_BHZ:2015}, where the harvested power of ERs can vary significantly depending on their distances from the ET. Thus, how to achieve a balanced performance among near-far ERs and yet with affordable low complexity is still challenging in implementing multi-user WPT systems in practice, especially with massive MIMO.

A promising low-complexity WPT scheme to achieve the above goal is the \emph{retrodirective} beamforming technique \cite{J_ZZC:2016}. In this technique, the ERs transmit a common beacon signal simultaneously to the ET in the uplink, and the ET simply amplifies the conjugate of the received sum-signal at each of its antennas and broadcasts to all ERs in the downlink for WPT. In this letter, we study a multi-user massive MIMO WPT system based on the technique of retrodirective EB. An important practical issue in retrodirective-based multi-user WPT is the so-called ``doubly near-far" problem \cite{J_ZZC:2016}, where the uplink beacon signal sent by a far ER is weakly received at the ET as compared to that of a near ER. Thus, the retrodirective EB becomes ineffective for the far ER in the downlink WPT, resulting in unfair performance among near-far ERs. To resolve this problem, we propose a distributed beacon power update algorithm, where each ER independently adjusts the beacon power based on its current harvested power in an iterative manner, subject to the maximum beacon power constraint. It is shown that with the proposed algorithm, the uplink beacon powers converge quickly to a unique fixed-point solution. Moreover, the converged beacon powers of the near ERs are effectively reduced as compared to those of the far ERs, thus helping to achieve the desired performance balance among all ERs with best efforts.  

\vspace{-2ex}
\section{System Model}
We consider a multi-user massive MIMO WPT system, where an ET equipped with $M_\text{t}$ antennas sends energy wirelessly to $K$ single-antenna ERs, where $M_\text{t} \gg K$. We assume a narrow-band block-fading channel model where the downlink channel from the $m$th antenna of the ET to ER$_k$ is denoted by $h_{km} = \sqrt{\beta_k}\tilde{h}_{km}$, $m=1,...,M_\text{t}$, $k=1,...,K$, with $\beta_k$ modelling the large-scale fading depending on the link distance, and $\tilde{h}_{km}$ representing the small-scale fading, such as the Rayleigh fading with $\tilde{h}_{km}$'s being independent and identically distributed (i.i.d.) circularly symmetric complex Gaussian (CSCG) random variables each with zero mean and unit variance, i.e., $\tilde{h}_{km}\sim \cC\cN(0,1)$, $m=1,...,M_\text{t}$, $k=1,...,K$; it is also assumed that $\beta_k$ remains constant in this letter but $\tilde{h}_{km}$ can vary from block to block. Denote the multiple-input single-output (MISO) downlink channel vector from the $M_\text{t}$ antennas of the ET to ER$_k$ as $\mv{h}_k^* \triangleq [h_{k1},...,h_{kM_\text{t}}]^T \in \mC^{M_\text{t}\times 1}$, where $\mv{a}^*$ and $\mv{a}^T$ denote the conjugate and transpose of a complex-valued vector $\mv{a}$, respectively. We further assume that the downlink and uplink channels are reciprocal, and thus the uplink channel vector from ER$_k$ to the ET is given by $\mv{h}_k^H$, where $\mv{a}^H$ denotes the conjugate transpose of a complex-valued vector $\mv{a}$.

\vspace{-1ex}
\section{Retrodirective WPT} 
\label{Section:RA}

In this section, we propose a low-complexity massive MIMO WPT scheme based on the technique of \emph{retrodirective} WPT by exploiting the channel reciprocity \cite{J_ZZC:2016}, where each block consists of two phases, as illustrated in Fig.~\ref{Fig:RA}.

\subsubsection{Beacon Phase}

\begin{figure}
\centering
\subfigure[Beacon phase]{
\centering
\includegraphics[width=4.2cm]{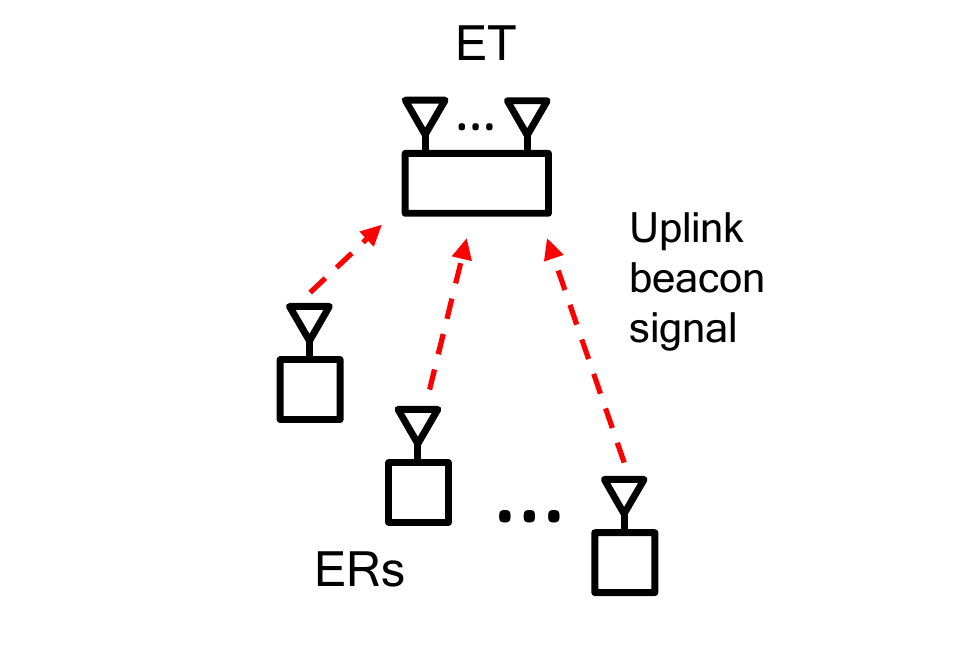}\label{Fig:RA1}} 
\subfigure[WPT phase]{
\centering
\includegraphics[width=4.2cm]{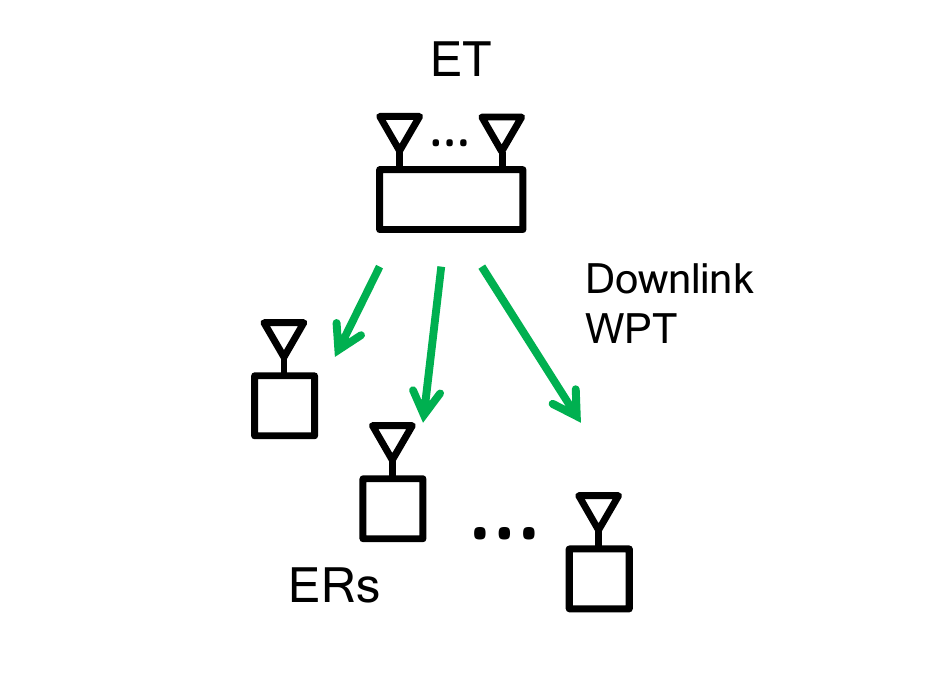}\label{Fig:RA2}} 
\caption{The proposed multi-user massive MIMO WPT system based on retrodirective transmission.}
\label{Fig:RA}
\vspace{-2ex}
\end{figure}

As shown in Fig.~\ref{Fig:RA1}, in the first phase for uplink training, referred to as the \emph{beacon phase}, each ER$_k$, $k=1,...,K$, sends a single-tone beacon waveform $\phi_k(t) =  \sqrt{2 p_k} \cos(2 \pi f_c t) $ in the uplink to the ET for a duration of $\tau > 0$, where $0 \leq p_k \leq \Pmax$ is the transmit power of ER$_k$  with $\Pmax$ denoting the maximum beacon power constraint, and $f_c$ is the carrier frequency. Thus, the equivalent baseband representation of the received signal at the ET from all ERs can be expressed as
\begin{align}
\mv{y}(t) & = \sum_{k=1}^K \sqrt{p_k} \mv{h}_k^* + \mv{z}(t) \\
& = \mv{g}  + \mv{z}(t), \; 0 \leq t \leq \tau,
\end{align}
where $\mv{z}(t) \triangleq [z_1(t),...,z_{M_\text{t}}(t)]^T$ denotes the i.i.d. zero-mean additive white Gaussian noise (AWGN) with power spectral density $N_0$, and $\mv{g} \triangleq \sum_{k=1}^K \sqrt{p_k} \mv{h}_{k}^*$ is the effective channel observed by the ET which is a weighted linear combination of the MISO channels of all the $K$ ERs. The ET then performs a matched-filter operation to its received signal $\mv{y}(t)$ to obtain $\hat{\mv{g}}$, which is given by
\begin{equation} \label{Eq:FilteredNoise}
\hat{\mv{g}}  = \frac{1}{\tau} \int_0^\tau \mv{y}(t) dt  = \mv{g} + \tilde{\mv{g}},
\end{equation}
where $\tilde{\mv{g}} \triangleq \frac{1}{\tau}\int_{0}^\tau \mv{z}(t) dt$. It can be shown that $\tilde{\mv{g}} \sim \cC\cN(\mv{0},\frac{N_0}{\tau} \mv{I})$, where $\mv{0}$ and $\mv{I}$ denote the all-zero vector with size $M_\text{t}\times 1$ and the identity matrix with size $M_\text{t}\times M_\text{t}$, respectively.

\subsubsection{WPT Phase}
Next, as shown in Fig.~\ref{Fig:RA2}, in the second phase for the downlink WPT, termed the \emph{WPT phase}, the ET transmits wireless energy to all ERs via the retrodirective WPT. Specifically, each antenna at the ET sends a single-tone sinusoidal waveform with the carrier frequency $f_c$ same as the uplink beacon signal, whose phase and amplitude are set according to the conjugate of the corresponding element in $\hat{\mv{g}}$, subject to the maximum total transmit power $P_\text{t}$ at the ET. Accordingly, the baseband equivalent of the transmitted signal vector from the $M_\text{t}$ transmit antennas of the ET can be expressed as
\begin{equation}
\mv{x} = \sqrt{P_\text{t}}\frac{\hat{\mv{g}}^*}{||\hat{\mv{g}}||},
\end{equation}
where we have dropped the time index $t$ since the baseband signal $\mv{x}$ is invariant over $t$. Then, the received signal at each ER$_k$ is given by $r_k = \mv{h}_k^H \mv{x}$, $k=1,...,K$. The corresponding harvested power at ER$_k$ during the WET phase, denoted by $Q_k(\mv{p})$, is a function of the beacon power vector $\mv{p} \triangleq [p_1,...,p_K]^T$, which can be expressed as
\begin{equation} \label{Eq:HarvestedPower_Retro} 
Q_k(\mv{p}) = \eta_k|r_k|^2 =  \eta_k\frac{P_\text{t}}{||\hat{\mv{g}}||^2} \l|\sum_{l=1}^K \sqrt{p_l}  \mv{h}_k^H \mv{h}_l  + \mv{h}_k^H\tilde{\mv{g}}^* \r|^2,
\end{equation}
where $0<\eta_k\leq 1$, $k=1,...,K$, denotes the RF-to-direct current (DC) energy conversion efficiency,  which is a constant and thus omitted in the sequel for brevity. It is interesting to note from \eqref{Eq:HarvestedPower_Retro} that the amount of harvested power at each ER$_k$ is related to its own as well as all other $K-1$ ERs' beacon powers.

\begin{lemma}
With $M_\text{t} \gg K$, the harvested power at each ER$_k$ given in \eqref{Eq:HarvestedPower_Retro} converges almost surely to
\begin{equation} \label{Eq:HarvestedPower_Retro_Massive}
Q_k(\mv{p}) \rightarrow  \l(P_\text{t} \beta_k+ P_\text{t}\frac{p_k \beta_k^2 (M_\text{t} -1)}{\sum_{l=1}^K p_l \beta_l + \frac{N_0}{\tau}}\r), \; k=1,...,K.
\end{equation}
\end{lemma}
\begin{proof}
With $M_\text{t} \gg K$, it can be shown that  $\frac{1}{M_\text{t}}||\mv{h}_k||^2 \rightarrow  \E[|h_{km}|^2] =   \beta_{k}$, $m=1,...,M_\text{t}$, $\frac{1}{M_\text{t}} \mv{h}_k^H \mv{h}_l \rightarrow  0$, $\forall k\neq l$, $\frac{1}{M_\text{t}}\mv{h}_k^H \tilde{\mv{g}}^* \rightarrow  0$, $\frac{1}{M_\text{t}} |\mv{h}_{k}^H \mv{h}_{l}|^2 \rightarrow  \beta_k \beta_l$, $\forall k \neq l$, $\frac{1}{M_\text{t}} |\mv{h}_{k}^H \tilde{\mv{g}}^*|^2 \rightarrow \beta_k \frac{N_0}{\tau}$, and $\frac{1}{M_\text{t}} ||\hat{\mv{g}}||^2 \rightarrow \l(\sum_{l=1}^K p_l \beta_l + \frac{N_0}{\tau}\r)$ \cite{J_M:2010}. Applying these results to \eqref{Eq:HarvestedPower_Retro} yields \eqref{Eq:HarvestedPower_Retro_Massive}. 
\end{proof}

Several interesting observations can be made from \eqref{Eq:HarvestedPower_Retro_Massive}. First, if $p_k = 0$, the harvested power at ER$_k$ is $Q_k(\mv{p}) = P_\text{t} \beta_k$. This means that regardless of the beacon power of all other $K-1$ ERs, even if ER$_k$ does not send any beacon signal, it can still harvest the constant power $P_\text{t} \beta_k$, which is the amount as if the ET isotropically broadcasts with power $P_\text{t}$. Second, for the case of $K=1$ and $p_1\beta_1 \gg \frac{N_0}{\tau}$, i.e., with only a single ER and sufficiently long duration of beacon phase $\tau$ or large $p_1\beta_1,$ where the noise effect can be ignored, \eqref{Eq:HarvestedPower_Retro_Massive} becomes $Q_1(p_1) = M_\text{t} P_\text{t} \beta_1$. This can be shown to be the maximum power that can be harvested with the optimal maximum ratio transmission (MRT) beamforming from the ET to ER$_1$, with its perfect CSI known at the ET, thus achieving the optimum single-user massive MIMO EB gain. Third, it can be seen from \eqref{Eq:HarvestedPower_Retro_Massive} that there exists a trade-off among the harvested powers at the $K$ ERs. Specifically, for each ER$_k$, the harvested power $Q_k(\mv{p})$ in \eqref{Eq:HarvestedPower_Retro_Massive} is a strictly increasing function of its own beacon power $p_k$, but a strictly decreasing function of all other $K-1$ ERs' beacon powers $p_l$, $\forall l \neq k$. In other words, using higher beacon power $p_k$ by ER$_k$ is helpful to increase its own harvested power, but will decrease the harvested powers of all other ERs. Last, if the same beacon power is used by all ERs, i.e., $p_1=...=p_K$, \eqref{Eq:HarvestedPower_Retro_Massive} will result in unfair performance among near-far ERs, since the dominant harvested power  of each ER (see the second term in \eqref{Eq:HarvestedPower_Retro_Massive}) is proportional to $\beta_k^2$. This reveals a critical ``doubly near-far" problem in retrodirective-based multi-user WPT; while in order to achieve the balanced harvested power for all ERs, it follows from \eqref{Eq:HarvestedPower_Retro_Massive} that ERs with smaller $\beta_k^2$ (i.e., more far-away from the ET) need to use significantly higher beacon power than those with larger $\beta_k^2$ (i.e., nearer to the ER), due to the round-trip power loss in both uplink beacon transmission and downlink WPT. 

It is worth noting that an effective method to mitigate the above doubly near-far problem is that the ER first obtains the knowledge of all $\beta_k$'s, $k=1,...,K$, via e.g., the reverse-link training \cite{J_ZZ:2015,J_ZZ:2015_b} and then assigns a proper value of $p_k$ to each ER$_k$ based on \eqref{Eq:HarvestedPower_Retro_Massive} to balance the harvested power among near-far ERs. However, this method requires multi-user orthogonal pilots assignment among ERs for uplink training as well as the downlink transmission of the optimized beacon powers to individual ERs, which is costly especially when the number of ETs/ERs is large. Thus, in the following section, we propose a distributed algorithm for each ER to independently update its beacon power in an iterative manner based only on its own received power and harvested power target, thus helping to achieve the desired user fairness  for the low-complexity retrodirective-based WPT.

\vspace{-1ex}
\section{Distributed Beacon-Power Update}

Let $\bar{Q}_k \geq 0$, $k=1,...,K$, denote the harvested power target that ER$_k$ aims to achieve, i.e.,  $Q_k(\mv{p}) \geq \bar{Q}_k$, where $Q_k(\mv{p})$ is given in \eqref{Eq:HarvestedPower_Retro_Massive}.  Without loss of generality, we assume that $\bar{Q}_k \geq P_\text{t}\beta_k$, $k=1,...,K$, since each ER$_k$ is guaranteed to harvest at least $P_\text{t}\beta_k$ amount of power even without transmitting a beacon signal, i.e., $p_k = 0$, as discussed in Section~\ref{Section:RA}. For convenience, we first define $q_k(\mv{p}) = Q_k(\mv{p}) - P_\text{t}\beta_k$ and $\bar{q}_k = \bar{Q}_k - P_\text{t}\beta_k$, $k=1,...,K$. It then follows that  $Q_k(\mv{p}) \geq \bar{Q}_k$ is equivalent to $q_k(\mv{p}) \geq \bar{q}_k$, $k=1,...,K$. We then propose the following beacon power update algorithm, where each ER$_k$ iteratively updates $p_k$ at the end of the WPT phase in each block based on the current value of $q_k(\mv{p})$ and the power target $\bar{q}_k$:
\begin{equation} \label{Eq:Iteration_Energy}
p_k[n+1] = \min\l\{\Pmax, \; \frac{\bar{q}_k}{q_k(\mv{p}[n])} p_k[n]\r\}, \; k=1,...,K, 
\end{equation}
where $n \geq 1$ is the block index. Note that for the above proposed scheme, each ER$_k$ only needs to know the value of $P_\text{t} \beta_k$, $k=1,...,K$. This can be easily obtained in an initial block ($n=0$) by each ER$_k$ via setting $p_k[0] = 0$ and then measuring the harvested power.

The power-update rule given in \eqref{Eq:Iteration_Energy} implies that each ER$_k$ will increase its beacon power  if $q_k(\mv{p}[n]) < \bar{q}_k$, provided that $0<p_k[n] < \Pmax$, or decrease it if $q_k(\mv{p}[n]) > \bar{q}_k$ and $p_k[n] > 0$. It is worth noting that \eqref{Eq:Iteration_Energy} resembles the distributed power control algorithm proposed in \cite{J_Y:1995,J_GZY:1994}, for the case of multi-user communication with co-channel interference. Then, the following theorem guarantees that for any given target $\bar{Q}_k \geq P_\text{t} \beta_k$, i.e., $\bar{q}_k \geq 0$, $k=1,...,K$, the beacon power update algorithm in \eqref{Eq:Iteration_Energy} converges to a unique fixed-point solution.
\begin{theorem} \label{Statement:RA WET Convergence}
For a given set of harvested power targets $\bar{Q}_k \geq P_\text{t} \beta_k$ and any initial beacon power vector $\mv{p}[1]$ with $0 < p_k[1] \leq \Pmax$, $k=1,...,K$, the beacon power update algorithm in \eqref{Eq:Iteration_Energy} converges to a unique fixed-point solution $\mv{p}^\star = [p_1^\star,...,p_K^\star]^T$, where the ERs with $p_k^\star < \Pmax$ have $Q_k(\mv{p}^\star) = \bar{Q}_k$, i.e., achieving their respective targets; whereas the ERs with $p_k^\star = \Pmax$ have $Q_k(\mv{p}^\star) \leq \bar{Q}_k$. 
\end{theorem}
\begin{proof}
The condition $Q_k(\mv{p}) \geq \bar{Q}_k$, $k=1,...,K$ (or equivalently $q_k(\mv{p}) \geq \bar{q}_k$), can be expressed in a matrix form:
\begin{equation}
\mv{p} \succeq \mv{A}(\mv{B}\mv{p} + \mv{\eta}),
\end{equation}
where $\mv{A} \triangleq \frac{1}{P_\text{t}(M_\text{t}-1)}\diag\l(\frac{\bar{q}_1}{\beta_1^2},..., \frac{\bar{q}_K}{\beta_K^2}\r)$, $\mv{B}\triangleq [\beta_1\mv{1},...,\beta_K\mv{1}]$ with $\mv{1} \triangleq [1,...,1]^T$ being the all-one vector of size $K\times 1$, and $\mv{\eta} \triangleq \frac{N_0}{\tau}\mv{1}$. Since the elements of both $\mv{A}$ and $\mv{B}$ are all positive, the convergence proof of the distributed constrained power control (DCPC) algorithm given in \cite{J_GZY:1994} can be directly applied to show that \eqref{Eq:Iteration_Energy} converges to a unique fixed-point solution starting from any feasible non-zero power vector $\mv{p}[1]$; the detailed proof is thus omitted due to the space limitation. 
\end{proof}

From \eqref{Eq:HarvestedPower_Retro_Massive} and Theorem~\ref{Statement:RA WET Convergence}, we obtain the following corollary for the converged beacon power solution.
\begin{corollary} \label{Statement:BestEffort}
Given a common harvested power target $\bar{Q}$ for all ERs, i.e., $\bar{Q}_k = \bar{Q}$, $k=1,...,K$, it holds that $p_k^\star \geq p_l^\star$, $k\neq l$, if and only if $\beta_k \leq \beta_l$. 
\end{corollary}

Corollary~\ref{Statement:BestEffort} implies that when $\bar{Q}_k = \bar{Q}$, $k=1,...,K$, the converged beacon power of a more far-away ER$_k$ (with smaller $\beta_k$) is no less than that of a nearer ER$_l$ (with larger $\beta_l$); as a result, the retrodirective-based WPT  becomes more effective for the farther ER$_k$, which thus helps mitigate the doubly near-far problem. However, if an ER$_k$ is too far away from the ET or its power target is too large,  it may not be able to achieve its target after the power-update algorithm converges; while in this case, from Theorem~\ref{Statement:RA WET Convergence}, it follows that $p_k^\star = \Pmax$, i.e., the proposed algorithm helps such ER to achieve its desired target with best efforts. 

\vspace{-2ex}
\section{Numerical Results} \label{Section:Simulation}
In this section, we present simulation results to validate the performance of our proposed beacon power update algorithm in \eqref{Eq:Iteration_Energy}. We set $M_\text{t} = 500$, $P_\text{t} = 1$ Watt (W), $\Pmax = 0.1$W, $f_c = 900$ MHz, $\tau = 10^{-6}$ second (s), $N_0 = -170$ dBm/Hz, and $\eta_k = 1$, $k=1,...,K$. Moreover, the large-scale channel attenuation $\beta_k$ is modelled as $\beta_k = c_0(r_k/r_0)^{-\alpha}$, where $c_0 = -30$ dB is a constant attenuation for the path-loss at a reference distance $r_0 = 1$ meter (m), $\alpha = 3$ is the path-loss exponent, and $r_k$ is the distance between the ET and ER$_k$.

We first present the convergence performance of the iterative algorithm \eqref{Eq:Iteration_Energy} in Figs.~\ref{Fig:Convergence1} and \ref{Fig:Convergence2} with $K=3$ and different values of the common harvested power target, i.e., $\bar{Q}_k = 0.1$ mW and $\bar{Q}_k = 0.24$ mW, $k=1,2,3$, respectively. Moreover, we set  $r_1 = 5$ m, $r_2 = 10$ m, $r_3 = 15$ m, and the initial beacon powers $p_k[1] = \Pmax = 0.1$ W, $k=1,2,3$. First, it can be observed from  both figures that the beacon power vector $\mv{p}[n]$ and the corresponding harvested power $Q_k(\mv{p}[n])$, $k=1,2,3$, converge to $\mv{p}^\star$ and $Q_k(\mv{p}^\star)$, respectively. Second, it can be seen from Fig.~\ref{Fig:Convergence1} that due to the relatively low  power target, we have $p_1^\star < p_2^\star < p_3^\star < \Pmax = 0.1$ W and $Q_k(\mv{p}^\star) = \bar{Q}_k$, $k=1,2,3$. By contrast, it is observed from Fig.~\ref{Fig:Convergence2} that with the increased power target for all ERs compared to that in Fig.~\ref{Fig:Convergence1}, we have $p_1^\star < p_2^\star < \Pmax = 0.1$ W, while $p_3^\star = \Pmax$. Consequently, $Q_k(\mv{p}^\star) = \bar{Q}_k$, $k=1,2$, whereas $Q_3(\mv{p}^\star) < \bar{Q}_3$ but at least it is improved as compared to the initial harvested power $Q_3(\mv{p}[1])$. The above results corroborate Theorem~\ref{Statement:RA WET Convergence} and Corollary~\ref{Statement:BestEffort}. Moreover, they confirm that the doubly near-far problem is effectively mitigated with best efforts by our proposed beacon power control algorithm, i.e., the converged beacon power of ER$_3$ (which is  most far-away from the ET) is no smaller than those of ER$_1$ and ER$_2$.

\begin{figure} 
\centering
\includegraphics[width=8.4cm]{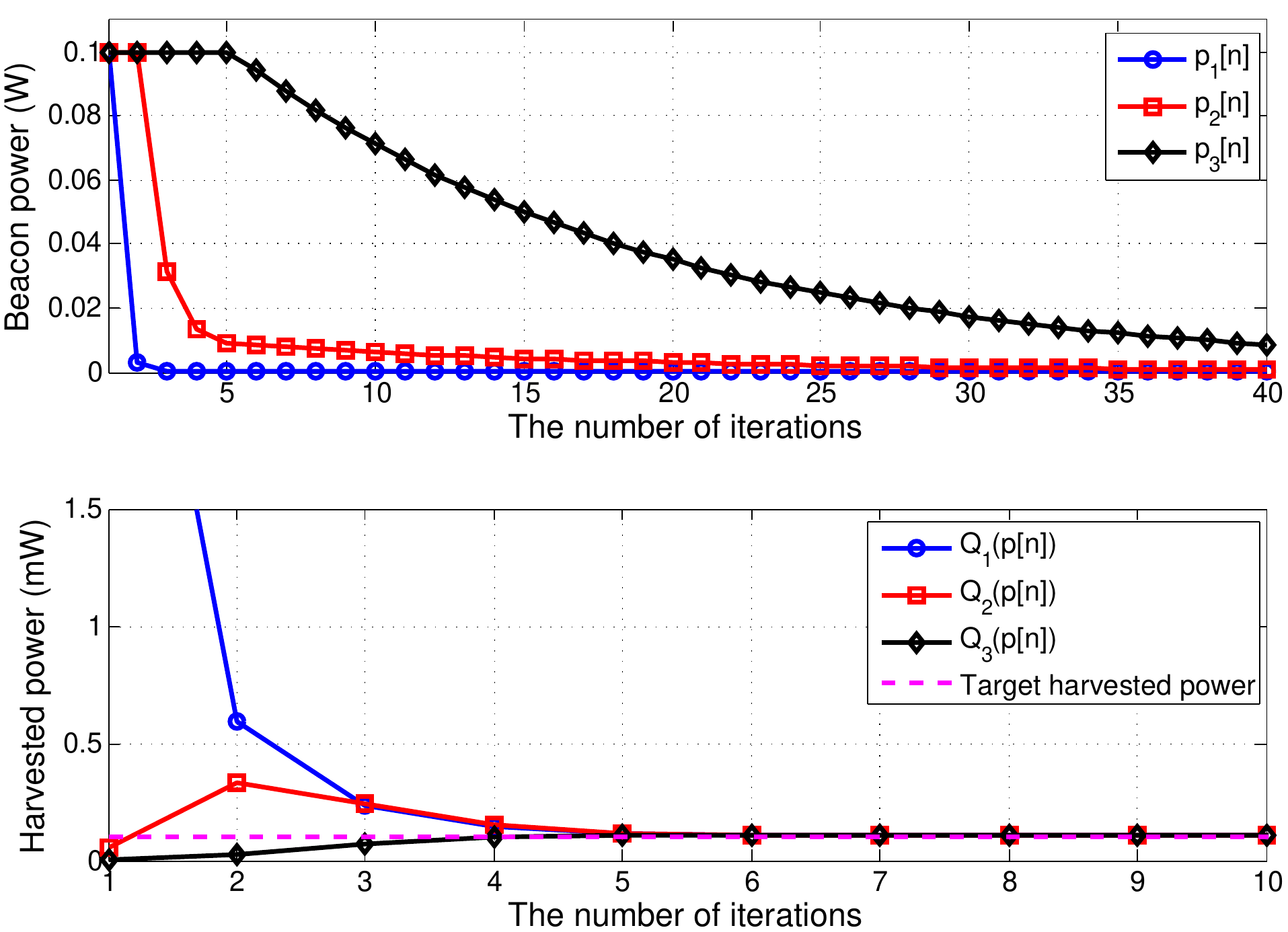}
\caption{Convergence of the beacon power and harvested power, where $p_k^\star < \Pmax = 0.1$ W, $k=1,2,3$. } \label{Fig:Convergence1}
\end{figure}

\begin{figure} 
\centering
\includegraphics[width=8.4cm]{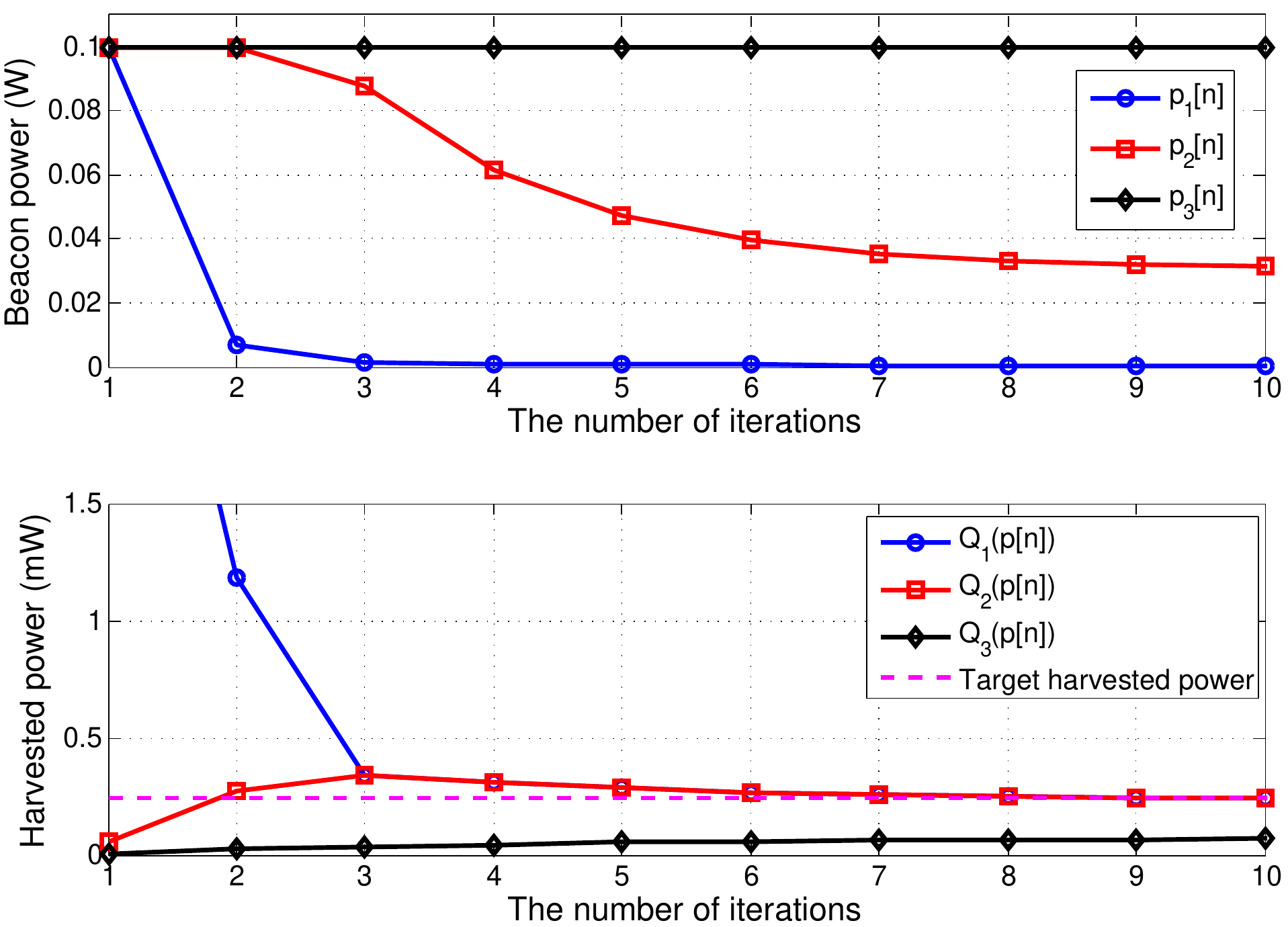}
\caption{Convergence of the beacon power and harvested power, where $p_1^\star,p_2^\star < \Pmax = 0.1$ W and $p_3^\star = \Pmax$. } \label{Fig:Convergence2} \vspace{-2ex}
\end{figure}

Last, in Fig.~\ref{Fig:Performance_Massive}, we plot the percentage of ERs that can achieve the common harvested power target $\bar{Q}_k = \bar{Q}$, $k=1,...,K$, after $20$ iterations of the proposed algorithm \eqref{Eq:Iteration_Energy} out of $K=30$ total ERs, by averaging over $5000$ randomly generated $r_k\sim \text{Uniform}(5\text{m},15\text{m})$, $k=1,...,K$.  For comparison, two benchmark schemes are also plotted, where each ER$_k$ fixes its beacon power to be $p_k = \Pmax$, $k=1,...,K$, and $p_k = 0.1\Pmax$, $k=1,...,K$, respectively. It can be observed from Fig.~\ref{Fig:Performance_Massive} that our proposed algorithm in general outperforms the benchmark schemes especially when the common power target $\bar{Q}$ is relatively small. As $\bar{Q}$ increases, the performance gain of the proposed algorithm diminishes, as in this case it is more likely that each ER cannot achieve its target even with maximum beacon power and thus the distributed beacon power control becomes less effective.  

\begin{figure} 
\centering
\includegraphics[width=8cm]{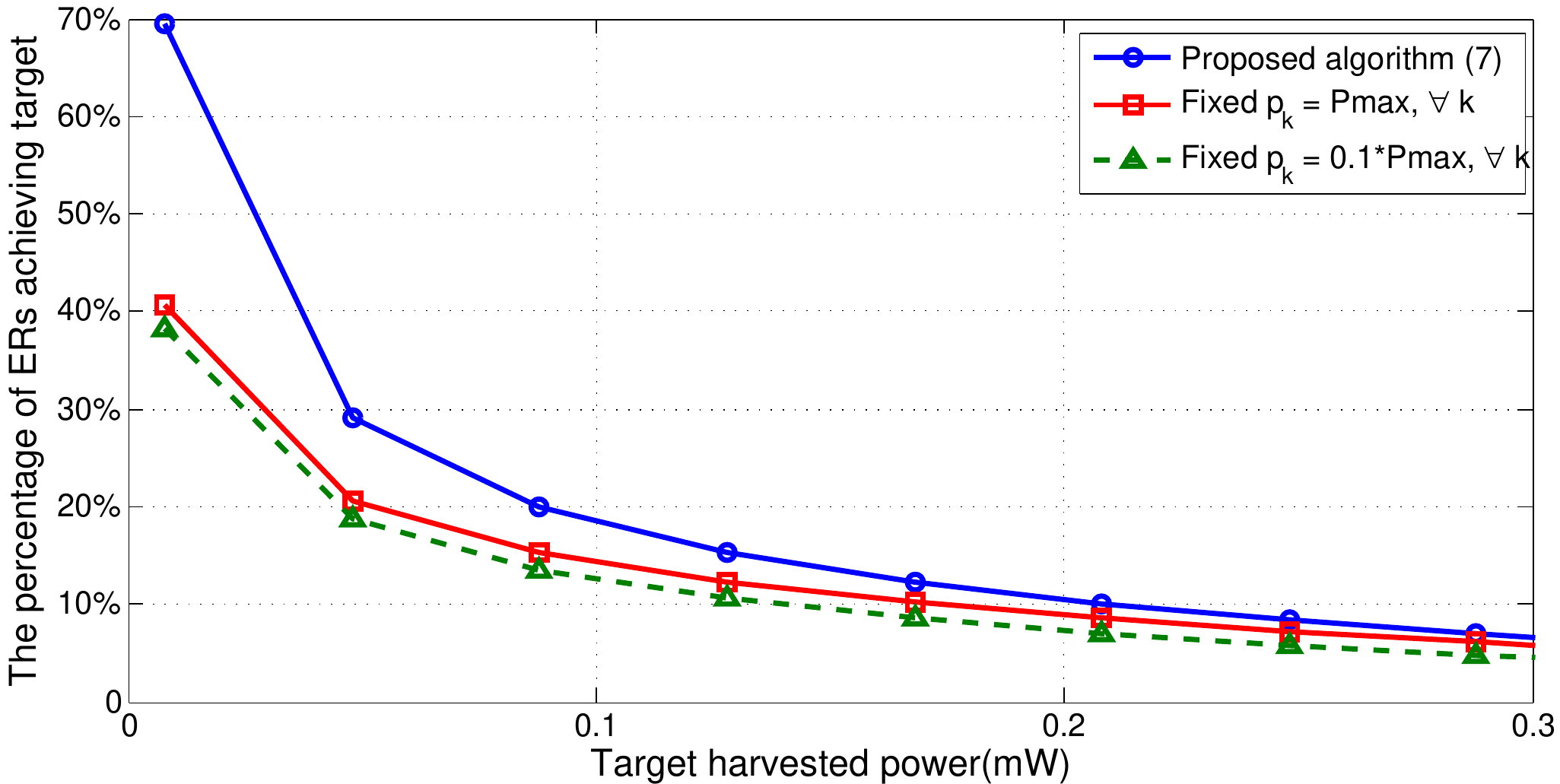}
\caption{The percentage of ERs that achieve the harvested power target with in total $K = 30$ ERs.} \label{Fig:Performance_Massive}
\vspace{-2ex}
\end{figure}

\vspace{-1ex}
\section{Conclusion}
In this letter, we investigated a low-complexity WPT scheme based on the new retrodirective beamforming technique in a multi-user massive MIMO WPT system. We proposed an efficient distributed uplink beacon power control algorithm where each ER independently updates its beacon power to achieve its harvested power target with best effort. It is shown that the proposed algorithm converges quickly to a fixed-point solution that achieves the desired performance balance among near-far ERs.



\newpage
\end{document}

%% file: header.tex
\newtheorem{theorem}{\underline{Theorem}}[section]
\newtheorem{lemma}{\underline{Lemma}}[section]
\newtheorem{corollary}{\underline{Corollary}}[section]

\newcommand{\mv}[1]{\mbox{\boldmath{$ #1 $}}}

\def\E{\mathsf{E}}

\def\l{\left}
\def\r{\right}
\def\({\left(}
\def\){\right)}

\def\b0{{\mathbf{0}}}

\def\mC{{\mathbb{C}}}

\def\cC{\mathcal{C}}

\def\cN{\mathcal{N}}

\newcommand{\diag}{\mathrm{diag}}

\def\Pmax{P_\mathsf{max}}